\documentclass[oribibl,UKenglish]{llncs}

\usepackage{amsmath}
\usepackage{amssymb}
\usepackage{graphicx}
\usepackage{subfigure}
\usepackage{algorithm}
\usepackage{algpseudocode}
\usepackage{color}
\usepackage{gensymb}
\usepackage{mathrsfs}
\usepackage{mathtools}
\usepackage{hyperref}

\usepackage{fancyhdr}
\makeatletter
\def\blfootnote{\xdef\@thefnmark{}\@footnotetext}
\makeatother

\newcommand{\removed}[1]{}

\newcommand{\ra}{\rightarrow}

\DeclarePairedDelimiter\floor{\lfloor}{\rfloor}

\title{On the Distributed Construction of Stable Networks in Polylogarithmic Parallel Time}

\author{Matthew Connor\inst{1} \and Othon Michail\inst{1} \and Paul G. Spirakis\inst{1,2}}
\institute{
Department of Computer Science, University of Liverpool, UK\\ \and Computer Engineering and Informatics Department, University of Patras, Greece\\
Email:\email{ M.Connor3@liverpool.ac.uk, Othon.Michail@liverpool.ac.uk, P.Spirakis@liverpool.ac.uk}}

\begin{document}

\maketitle

\begin{abstract}
We study the class of networks which can be created in polylogarithmic parallel time by
\emph{network constructors}: groups of anonymous agents that interact randomly under a uniform random scheduler with the ability to form connections between each other.
Starting from an empty network, the goal is to construct a stable network which belongs to a given family.
We prove that the class of trees where each node has any $k \geq 2$ children can be constructed in $O(\log{n})$ parallel time with high probability.
We show that constructing networks which are $k$-regular is $\Omega(n)$ time, but a minimal relaxation to $(l, k)$-regular networks, where $l = k - 1$ can be constructed in polylogarithmic parallel time for any fixed $k$, where $k > 2$.
We further demonstrate that when the finite-state assumption is relaxed and $k$ is allowed to grow with $n$, then $k = \log\log{n}$ acts as a threshold above which network construction is again polynomial time.
We use this to provide a partial characterisation of the class of polylogarithmic time network constructors.\newline\\
\textbf{Submitted to Track B: Distributed and Mobile}\newline\\
\noindent\textbf{All omitted details are included in a clearly marked Appendix, 
to be read at the discretion of the Program Committee.}
\end{abstract}

\noindent
\textbf{Keywords:} population protocol, distributed network construction, polylogarithmic time protocol, spanning tree, regular network, partial characterisation

\section{Introduction}
\label{sec:intro}

\emph{Passively dynamic networks} are an important type of dynamic networks in which the network dynamics are \emph{external} to the algorithm and are a property of the environment in which a given system operates. Wireless sensor networks in which individual sensors are carried by autonomous entities, such as animals, or are deployed in a dynamic environment such as the flow of a river are examples of passively dynamic networks. In terms of modelling such systems, the network dynamics are usually assumed to be controlled by an \emph{adversary scheduler}, who has exclusive control on the interaction or communication sequence among the computational entities. 

One line of research has been assuming the scheduler to be \emph{fair}, in the sense that it can forever conceal potentially reachable configurations of the system. This sub-type of passively dynamic networks are known as population protocols and were introduced in the seminal paper of Angluin \emph{et al.} \cite{AADFP06}. \footnote{Which, by the way, was this year's recipient of the Edsger W. Dijkstra Prize in Distributed Computing.} A type of fair scheduler which is typically assumed when the running time of protocols is to be analysed, is the \emph{uniform random scheduler}, which in every discrete step selects equiprobably a pair of entities to interact from all permissible pairs of entities. Traditionally, the population protocols literature had been considering extremely weak entities and the goal was to reveal the computational possibilities and limitations under such a challenging interaction scheme. Recent progress has been highlighting the interesting trade-offs between local space of the entities and the running time of protocols, showing among other things that very fast running times (where fast is here considered to be anything growing as polylog$(n)$, $n$ being the total number of entities in the system) can be achieved for a wide range of basic distributed tasks if the entities are equipped with as few states as polylog$(n)$. Alistarh and Gelashvili \cite{AG15} have also proposed the first sub-linear leader election protocol, which stabilizes in $O(\log^3{n})$ parallel time, assuming $O(\log^3{n})$ states at each agent. Gasieniec and Stachowiak \cite{GS18} designed a space optimal ($O(\log\log{n})$ states) leader election protocol, which stabilises in $O(\log^2{n})$ parallel time. General characterizations, including upper and lower bounds of the trade-offs between time and space in population protocols are provided in \cite{AAEGR17}. Doty \emph{et al.} \cite{DEMST18} show that a state count of $O(n^{60})$ enables fast and exact population counting.

Another line has been considering worst-case adversary schedulers, which may even be aware of the protocol and trying to optimise against it. There, the entities are typically assumed to be powerful, like processors of traditional distributed systems, and the only restrictions imposed on the scheduler are instantaneous or temporal connectivity restrictions which essentially do not allow the scheduler to forever block communication between any two parts of the system. This was initiated by O'Dell and Wattenhofer \cite{OW05} for the asynchronous case and then the synchronous case was extensively studied in a series of papers by Kuhn \emph{et al.} \cite{KLO10}. Michail \emph{et al.} \cite{MCS12b} extended this to the case of possibly disconnected dynamic networks, in which connectivity is only guaranteed in a temporal sense.  

The other main type of dynamic networks with respect to who controls the changes in the network topology, are \emph{actively dynamic networks}. In such networks, the algorithm is able to either implicitly change the sequence of interactions by controlling the mobility of the entities or explicitly modify the network structure by creating and destroying communication links at will. This is for example the subject of the area of overlay network construction \cite{AACW05,AS07,AW07,GHS19} and very recently Michail \emph{et al.} introduced a fully distributed model for computation and reconfiguration in actively dynamic networks \cite{MSS20}.

An interesting alternative family of dynamic networks rises when one considers a mixture of the passive network dynamics of the environment and the active dynamics resulting from an algorithm that can partially control the network changes or that can fix network structures that the environment is unable to affect. This is naturally motivated by molecular interactions where, for example, proteins can bind to each other, forming structures and maintaining their stability  despite the dynamicity of the solution in which they reside. Michail and Spirakis \cite{MS16a} introduced and studied such an abstract model of distributed network construction, called the \emph{network constructors} model, where the network dynamicity is the same as in population protocols but now the finite-state entities can additionally activate and deactivate pairwise connections upon their interactions. It was shown that very complex global networks can be formed stably despite the dynamicity of the environment. Then Michail \cite{Mi18} studied a geometric variant of network constructors, in which the entities can only form geometrically constrained shapes in 2D or 3D space. Another interesting hybrid dynamic network model is the one by Gmyr \emph{et al.} \cite{GHSS17}, in which the entities have partial control over the connections of an otherwise worst-case passively dynamic network, following the model of Kuhn \emph{et al.} \cite{KLO10}.

\subsection{Our Approach}
\label{subsec:approach}

We investigate which families of networks can be stably constructed by a distributed computing system in polylogarithmic parallel time. To our knowledge, this is the first attempt made to approach this task.

Our protocols assume the existence of a leader node. A node $x$ is a \emph{leader node} if in the initial configuration all $u \in V \setminus \{x\}$, where $V$ is the set of all nodes, are in state $q_0$ and $x$ is in state $s \neq q_0$.

We first study the \emph{k-Children Spanning Tree} problem, where the goal is to construct a tree where each node has at most $k \geq 2$ children. We show that it is possible to solve this problem for any $k$ in $O(\log{n})$ time with high probability.
We then show that network constructors which create $k$-regular graphs necessarily take $\Omega(n)$ time. However, with a minimal relaxation to $(k,k - 1)$-regular networks the problem can be solved for any constant $k \geq 2$ in polylogarithmic time.
We examine this as a special case of the \emph{$(l,k)$-Regular Network} problem, where the goal is to construct a spanning network in which every node has at least $l < k$ and at most $k$ connections, where $2 < k < n$.
We then transitioned to experimental analysis of the protocol which not only provided evidence of the sharp contrast of the minimal relaxation but also revealed a threshold value for $k$, beyond which the problem reverts to polynomial time.
We used this knowledge to propose a first partial characterisation of the set of polylogarithmic time network constructors.
We leave providing formal bounds as an open problem, with a potential proof strategy provided in the Appendix.

In Section 2, we formally define the model of network constructors and the network construction problems that are considered in this work.
In Section 3, we study the $k$-children spanning tree problem, first for $k = 2$, and then for $k \geq 2$.
In Section 4, we first provide the lower bound for $k$-regular networks. We then present a protocol for the $(l, k)$-regular network problem and our experimental analysis culminating in the partial characterisation.
In Section 5, we conclude and give further research directions that are opened by our work.

\section{Preliminaries and Definitions}
\label{sec:prel}

\subsection{The model}
\label{subsec:model}

\begin{definition}
\normalfont A \emph{Network Constructor} (NET) is a distributed protocol defined by a 4-tuple $(Q, q_0, Q_{out}, \delta)$, where $Q$ is a finite set of \emph{node-states}, $q_0 \in Q$ is the \emph{initial node-state}, $Q_{out} \subseteq Q$ is the set of \emph{output node-states}, and $\delta:Q \times Q \times \{0,1\} \rightarrow Q \times Q \times \{0,1\}$ is the \emph{transition function}.
\end{definition}

If $\delta(a,b,c)=(a',b',c')$, we call $(a,b,c) \rightarrow (a',b',c')$ a \emph{transition} (or \emph{rule}) and we define $\delta_1 (a,b,c)=a', \delta_2(a,b,c)=b'$, and $\delta_3(a,b,c)=c'$.
A transition $(a,b,c) \rightarrow (a',b',c')$ is called \emph{effective} if $x \neq x'$ for at least one $x \in \{a,b,c\}$ and \emph{ineffective} otherwise.
When we present the transition function of a protocol we only present the effective transitions. Additionally, we agree that the \emph{size} of a protocol is the number of its states, i.e., $|Q|$.

The system consists of a population $V_I$ of $n$ distributed \emph{processes} (called \emph{nodes} for the rest of this paper).
In the generic case, there is an underlying \emph{interaction graph} $G_I=(V_I,E_I)$ specifying the permissible interactions between the nodes.
Interactions in this model are always pairwise.
In this work, $G_I$ is a \emph{complete undirected interaction graph}, i.e., $E_I=\{uv: u,v \in V_I \text{and } u \ne v\}$, where $uv=\{u,v\}$. Initially, all nodes in $V_I$ are in the initial node-state $q_0$.
A central assumption of the model is that edges have binary states. An edge in state 0 is said to be \emph{inactive} while an edge in state 1 is said to be \emph{active}. All edges are initially inactive.
Execution of the protocol proceeds in discrete steps. 
In every step, a pair of nodes $uv$ from $E_I$ is selected by an \emph{adversary scheduler} and these nodes interact and update their states and the state of the edge joining them according to the transition function $\delta$.

A \emph{configuration} is a mapping $C:V_I \cup E_I \rightarrow Q \cup \{0,1\}$ specifying the state of each node and each edge of the interaction graph.
Let $C$ and $C'$ be configurations, and let $u,\upsilon$ be distinct nodes.
We say that $C$ goes to $C'$ via \emph{encounter} $e=u\upsilon$, denoted $C \xrightarrow{e} C'$, if $(C'(u),C'(v), C'(e)) = \delta(C(u),C(v), C(e))$ or $(C'(v),C'(u),C'(e)) = \delta(C(v),C(u),C(e))$ and $C'(z)=C(z)$, for all $z \in (V_I \setminus \{u,v\}) \cup (E_I \setminus \{e\})$.
We say that $C'$ \emph{is reachable in one step from} $C$, denoted $C \rightarrow C'$, if $C \xrightarrow{e} C'$ for some encounter $e \in E_I$.
We say that $C'$ is \emph{reachable} from $C$ and write $C \leadsto C'$, if there is a sequence of configurations $C=C_0, C_1 ,..., C_t = C'$, such that $C_i \rightarrow C_{i+1}$ for all $i, 0 \leq i < t$.

An \emph{execution} is a finite or infinite sequence of configurations $C_0, C_1, C_2,...,$ where $C_0$ is an initial configuration and $C_i \rightarrow C_{i+1}$, for all $i \geq 0$.
A \emph{fairness condition} is imposed on the adversary to ensure the protocol makes progress. An infinite execution is \emph{fair} if for every pair of configurations $C$ and $C'$ such that $C \rightarrow C'$, if $C$ occurs infinitely often in the execution then so does $C'$.
In what follows, every execution of a NET will by definition considered to be fair.

We define the \emph{output of a configuration} $C$ as the graph $G(C)=(V,E)$ where $V=\{u \in V_I : C(u) \in Q_{out}\}$ and $E = \{uv : u,v \in V, u \neq v, \text{and } C(uv)=1\}$.
In words, the output-graph of a configuration consists of those nodes that are in output states and those edges between them that are active, i.e., the active subgraph induced by the nodes that are in output states.
The output of an execution $C_0, C_1, ... $ is said to \emph{stabilize} (or \emph{converge}) to a graph $G$ if there exists some step $t \geq 0$ such that (abbreviated “s.t.” in several places) $G(C_i)=G$ for all $i\geq t$, i.e., from step $t$ and onwards the output-graph remains unchanged. 
Every such configuration $C_i$, for $i \geq t$, is called \emph{output-stable}.
The \emph{running time} (or \emph{time to convergence}) of an execution is defined as the minimum such $t$ (or $\infty$ if no such $t$ exists).
Throughout the paper, whenever we study the running time of a NET, we assume that interactions are chosen by a \emph{uniform random scheduler} which, in every step, selects independently and uniformly at random one of the $|E_I| = n(n-1)/2$ possible interactions.
In this case, the running time becomes a random variable (abbreviated “r.v.” throughout) $X$ and our goal is to obtain bounds on the expectation $E[X]$ of $X$. Note that the uniform random scheduler is fair with probability 1.

In this work “time”  is treated as sequential in our analyses, i.e., a time-step consists of a single interaction selected by the scheduler. Such a sequential estimate can be easily translated to some estimate of parallel time. For example, assuming that $\Theta(n)$ interactions occur in parallel in every step, one could obtain an estimation of parallel time by dividing sequential time by $n$. All results are given in parallel time.

\begin{definition}
\normalfont We say that an execution of a NET on $n$ nodes \emph{constructs a graph} (or \emph{network}) $G$, if its output stabilizes to a graph isomorphic to $G$.
\end{definition}

\begin{definition}
\normalfont We say that a protocol $P$ \emph{constructs a graph language} $\mathscr{G}$, if in every execution $P$ constructs a graph $G \in \mathscr{G}$ and for all $G$ there exists an execution of $P$ which constructs $G$.
\end{definition}

\subsection{Problem definitions}
\label{subsec:problems}

Here we provide formal definitions for all of the classes of networks considered in this paper.\\

\noindent\\
\emph{k-Children Spanning Tree}. The goal is to construct a spanning tree where each individual element has at most $k \in \mathbb{N}$ children.\\
\noindent\\
\emph{$(l,k)$-Regular Network}. A spanning network where for any $l, k \in \mathbb{N}$ where $l < k$, elements with degree $d < l$ form a clique and all others have a degree of at least $l$ and at most $k$.\\

\subsection{Experimental Setup}
\label{subsec:experiment-setup}

We performed experiments with the goal of guiding a proof of the running time necessary to solve the $(l, k)$-regular network problem.
We learned that a formal proof would be difficult due to the reliance of random variables on the values of other random variables, so we leave this as an open problem.
We then experimented with different values of $k$ to see what the effect would be, and discovered a running time threshold in the process.
All were implemented using C and compiled with GCC.
All tests were repeated at least five times per value of $n$ and the average number of time steps taken as the result.
To terminate our experiments we designed special stabilisation conditions.
Details including a formal proof of correctness can be found in the Appendix.

\section{Polylogarithmic-time Protocols for $k$-Children Spanning Tree}

In this section, we study the complexity of the $k$-Children Spanning Tree problem. 
We first focus on the special case where  $k = 2$ and give a protocol (Protocol 1).
We show that it has a running time of $O(\log{n})$ parallel time with high probability.
Finally, we generalise for all $k \geq 2$  by giving a protocol (Protocol 2) and prove that the running time is again $O(\log{n})$.
\subsection{2-Children Spanning Tree}
\label{subsec:2-children-tree}

\vspace{-5mm}

\floatname{algorithm}{Protocol}
\renewcommand{\algorithmiccomment}[1]{// #1}
\begin{algorithm}[!h]
\begin{minipage}[b]{0.62\textwidth}
  \caption{\emph{2-Slot protocol}}\label{prot:gline}
  \begin{algorithmic}
    \medskip
    \State $Q=\{F, L_0,L_1,L_2,O_0,O_1,O_2\}$
    \State $\delta$:
    \begin{align*}
    (L_0, F, 0) &\ra (L_1, O_0, 1)\\
    (L_1, F, 0) &\ra (L_2, O_0, 1)\\
    (O_0, F, 0) &\ra (O_1, O_0, 1)\\
    (O_1, F, 0) &\ra (O_2, O_0, 1)
    \phantom{\hspace{10cm}}
    \end{align*}
    \State \Comment {All transitions that do not appear have no effect}
  \end{algorithmic}
  \end{minipage}
\end{algorithm}

\vspace{-5mm}

In the above protocol, the $F$ state corresponds to being a node which is not a member of the tree.
$L_i$ corresponds to the \emph{leader} node which acts as the root of the tree, and $O_i$ to non-leader nodes in the tree, where $i$ represents the number of children of a given node.
We assume that for every execution of Protocol 1 on a population $P$ of $n$ nodes, $n-1$ nodes initialise to the state $F$ and one node initialises to the state $L_0$.

\begin{lemma}
Protocol 1 stably constructs the graph language $\mathscr{T} = \{G | G \text{ is a tree and } \forall u \in P \implies \Delta^+(u) \leq 2\}$,
 where $\Delta^+(u)$ is defined as the number of children of the node $u$ in $O(\log{n})$ parallel time.
\end{lemma}

\begin{proof}
A full proof of this theorem is located in the Appendix.
\qed
\end{proof}

\begin{lemma}
For each time step in Protocol 1, the probability of any node in the set of unconnected nodes $U$ connecting to the tree is at least $2\frac{|U|}{n}$.
\end{lemma}

\begin{proof}
Assume there are $|S|$ nodes which are connected to the tree. The probability of a node $x \in U$ connecting to the tree is $\frac{|S||U|}{n^2}$.
If there are at least $n/2$ nodes connected to the tree, then $\frac{|S||U|}{n^2} \geq \frac{1/2|U|}{n^2} = 2\frac{|U|}{n}$.
The case where there are less than $n/2$ nodes connected to the tree is symmetrical, meaning that same process happens in reverse for $1 \leq n \leq n/2$.
Therefore $2\frac{|U|}{n}$ is a lower bound of the probability of connecting to the tree.
\qed
\end{proof}

\begin{lemma}
For Protocol 1, the number of time steps until convergence is $O(\log{n})$ w.h.p.
\end{lemma}

\begin{proof}
Consider the scenario where $m$ balls are being thrown into $n$ bins.
If $Z$ is the random variable for the number of empty bins, then $E[Z] = n(1 - 1/n)^m$.
The probability of a ball entering an empty bin is $e/n$, there e is the number of empty bins.
For our protocol scenario, balls are time steps and bins are unconnected nodes.
So $m = an\ln{n}$ is the number of balls, and $n$ is the number of bins.
Since the probabilty of success in the balls and bins scenario is lower than $2\frac{|U|}{n}$ when $|U| = e$, we can use it as a bound for the probability of success.
Therefore, $E[Z] = n(1 - 1/n)^{an\ln{n}} \leq  ne^{-an\ln{n}} = n^{1-a}$.
Using Markov's inequality, $E[Z \geq 1] \leq E[Z] = \frac{1}{n^a}$.
Since $a$ can be set arbitrarly high, convergence is O($\log{n}$) w.h.p.
\qed
\end{proof}

\begin{theorem}
Protocol 1 stably constructs the graph language $\mathscr{T}$ in $O(\log{n})$ time w.h.p.
\end{theorem}

\begin{proof}
By application of Lemmas 1 and 3.
\qed
\end{proof}

\subsection{$k$-Children Spanning Tree}
\label{sec:k-spanning-tree}

We now consider the problem of constructing the graph language $\mathscr{T}_k= \{G | G \text{ is a rooted tree and }\forall u \in P \implies \Delta^+(u) \leq k\}$.
Protocol 2 below operates in the same way as Protocol 1 but relaxes the finite-state restriction to provide states and rules for all $i \leq k$, where $k \geq 2$.

\vspace{-5mm}

\floatname{algorithm}{Protocol}
\renewcommand{\algorithmiccomment}[1]{// #1}
\begin{algorithm}[H]
\begin{minipage}[b]{0.62\textwidth}
  \caption{\emph{k-Slot protocol}}\label{prot:gline}
  \begin{algorithmic}
    \medskip
    \State $Q=\{F, L_0, L_1, \ldots, L_k, O_0, O_1, \ldots, O_k\}$
    \State $\delta$:
    \begin{align*}
    (L_x, F, 0) &\ra (L_{x+1}, O_0, 1) \text{ for } x < k\\
    (O_y, F, 0) &\ra (O_{y+1}, O_0, 1) \text{ for } y < k
    \phantom{\hspace{10cm}}
    \end{align*}
  \end{algorithmic}
  \end{minipage}
\end{algorithm}

\vspace{-5mm}

\begin{lemma}
Under Protocol 2, the connected component $S$, defined as the leader node and all nodes connected to the leader either directly or indirectly through some other nodes is eventually spanning.
\end{lemma}

\begin{proof}
We observe that the number of open slots $o$ is initially $k$.
$o$ is non-decreasing, as every increase in $\Delta^+(u)$ for some $u$ necessarily increases $|V(S)|$.
Since there are always open slots available, every unconnected node is guaranteed to be able to connect to $S$ at some point.
Therefore when $S$ stabilises it will contain all $u \in P$.
\qed
\end{proof}

\begin{lemma}
For all executions of Protocol 2 on the population $P$ of n nodes, it stabilizes to some $G \in \mathscr{T}_k$ where $|V(G)| = n$.
\end{lemma}

\begin{proof}
We prove this via an induction on the connected component $S$.
For the base case, there is one node in the state $L_0$. This is trivially a member of $\mathscr{T}_k$ as no connections have formed yet. We now assume that there is a connected component of size $|S|$.
For a connected component of size $|S| + 1$, an unconnected node $u \in V\setminus{S}$ in the state $F$ must connect to $S$ at some node $x \in S$.
By Lemma 3, such a node must exist. If the node $x$ has two children it is in the state $O_2$ or $L_2$, as for all nodes in states $O_i$ and $L_j$ the $i$ and $j$ correspond to the number of children of those nodes.
Since there is no defined transitions from these states no $u$ can connect to $x$.
\noindent
Therefore $S$ remains a tree and $G(S) \in \mathscr{T}_k$.
\qed
\end{proof}

\begin{lemma}
For all $G \in \mathscr{T}_k$, there is an execution of Protocol 2 which stabilises on $G$ when starting on a population $P$ of size $n = |V(G)|$.
\end{lemma}

\begin{proof}
We first set the value of $k$ to the maximum number of connections in any node in the tree.
Let the leader node $l$ in the population $P$ correspond to the root $r$ of $G$. If $r$ has $i$ children, connect $i$ nodes in the state $F$ to $l$.
For each child $c$ of the leader node, let it correspond to a child $d$ of $r$. If $d$ has $j$ children, connect $j$ nodes in the state $F$ to $c$.
Continuing this process for all nodes $u \in G$, the result is a spanning tree where all nodes in the tree are equivalent to some $u \in G$.
\qed
\end{proof}

\begin{theorem}
Protocol 2 stably constructs the graph language $\mathscr{T}_k$ in $O(\log{n})$ time w.h.p.
\end{theorem}

\begin{proof}
By application of the Lemmas above.
Protocol 2 can only be faster than Protocol 1 as it has more open slots per node.
\qed
\end{proof}

\section{Time Thresholds for $(l,k)$-Regular Networks}
\label{sec:polylog}

In this section, we present our solution for the $(l,k)$-Regular Network problem for $l = k - 1$, the \emph{Cross-edges Tree} protocol.
We first show that a \emph{$k$-regular network}, defined as a network where each node has degree exactly equal to $k$, cannot be constructed in polylogarithmic time.
We then show via experimental analysis that this impossibility result does not hold for the minimal relaxation of $(l,k)$-Regular Networks when $k$ is a constant and $l = k - 1$.
Finally, we demonstrate that when $k$ exceeds the threshold of $\log\log{n}$, the protocol itself is no longer in the polylogarithmic time class.
Note that from now on $k$ refers to the \emph{degree} of a node, not the \emph{number of children}.

\begin{theorem}
Any protocol which constructs a $k$-regular network where $k < n$ has a running time of $\Omega(n)$. 
\end{theorem}

\begin{proof}
Consider the population $P$ of size $n$ using a generic $k$-regular network construction protocol $X$. The number of connections is limited by $k$ to $\frac{kn}{2}$ as this is less than the $\frac{n(n-1)}{2}$ maximum for $n$ nodes. The population initially has $kn$ network connection entry points which can be used to make new connections and which decrease by 2 for every connection made. Since $(kn) \leq n(n-1)$, at some point in the execution there must be two nodes with 1 unused entry point each. Using these points and stabilising the protocol means both nodes must be selected by the scheduler at the same time, an event with probability $\frac{1}{n^2}$. Since an event with probability $\frac{1}{n^2}$ is unavoidable the protocol $X$ must construct a network in at least $\Omega(n^2)$ interactions. 
\qed
\end{proof}

In light of the above impossiblity, we now give our protocol for the $(l,k)$-Regular Network problem when $l = k - 1$.

\vspace{-5mm}

\floatname{algorithm}{Protocol}
\renewcommand{\algorithmiccomment}[1]{// #1}
\begin{algorithm}[H]
\begin{minipage}[b]{0.62\textwidth}
  \caption{\emph{Cross-edges Tree}}\label{prot:gline}
  \begin{algorithmic}
    \medskip
    \State $Q=\{F, L_0, L_1, \ldots, L_k, O_0, O_1, \ldots, O_k\}$
    \State $\delta$:
    \begin{align*}
    (L_x, F, 0) &\ra (L_{x+1}, O_0, 1) \text{ for } x < k\\
    (O_y, F, 0) &\ra (O_{y+1}, O_0, 1) \text{ for } y < k\\
    (L_x, O_y, 0) &\ra (L_{x+1}, O_{y+1}, 1) \text{ for } x, y < (k - 1)\\
    (O_y, O_z, 0) &\ra (O_{y+1}, O_{z+1}, 1) \text{ for } y, z < (k - 1)
    \phantom{\hspace{10cm}}
    \end{align*}
  \end{algorithmic}
  \end{minipage}
\end{algorithm}

\vspace{-5mm}

The Cross-edges Tree protocol adds additional rules allowing leaves within a tree to connect to other nodes within the tree as though they are candidates for becoming children.

We now provide the results of simulating the protocol for $k = 3$.
We used the same conditions as in the other running time experiments, executing the protocol 10 times for each population size $n$, where $n = 10 + 6t$, where $t$ is the test number from 0 to 199.

\begin{figure}[h!]
\includegraphics[width=\textwidth]{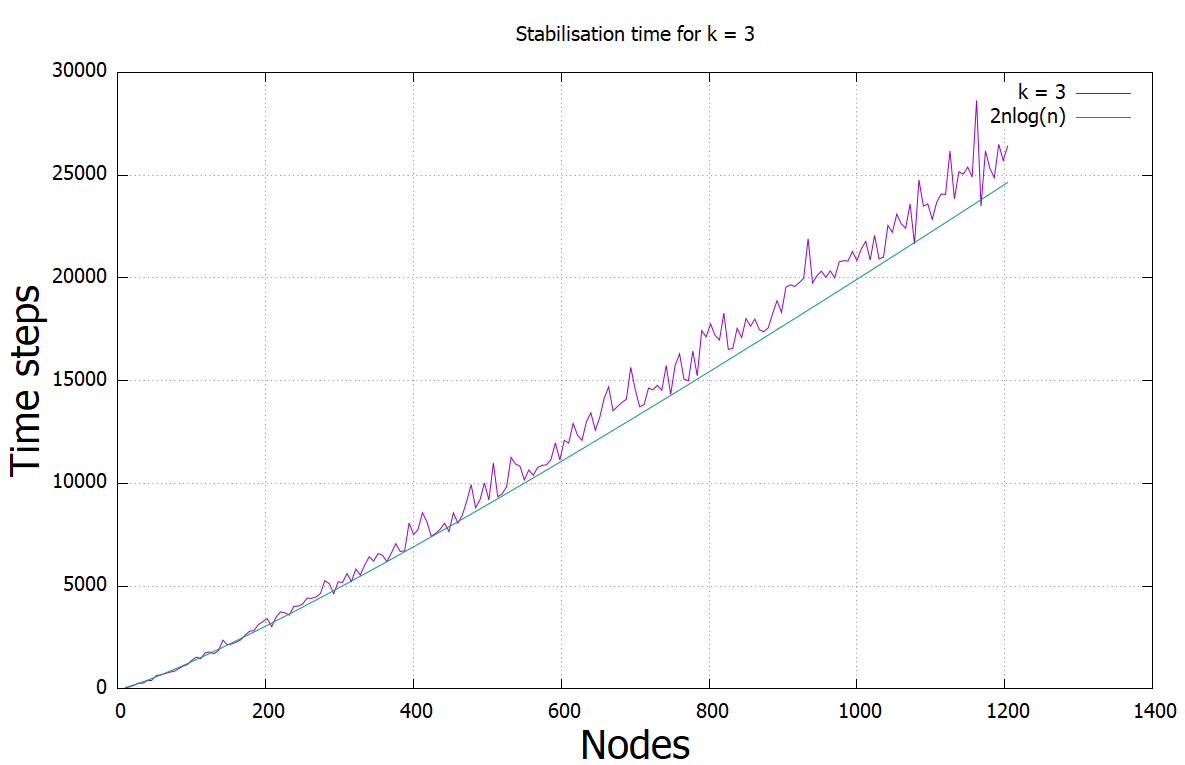}
\caption{Running time of the protocol for k = 3, compared with a polylogarithmic function.} \label{fig1}
\end{figure}

The running time is difficult to prove formally.
This is because random variables are used which represent the number of nodes with a given degree in a given time step.
Their values depend on the values of all random variables in the previous time step.
We therefore turn our focus to experiments based on measuring the impact of the value of $k$ on the running time of the protocol.

We have measured the running time of our Cross-edges Tree protocol for different network sizes.
The results below show that a higher value of $k$ has little effect on the running time until $k$ exceeds $\log\log{n}$.

\begin{figure}[H]
\includegraphics[width=\textwidth]{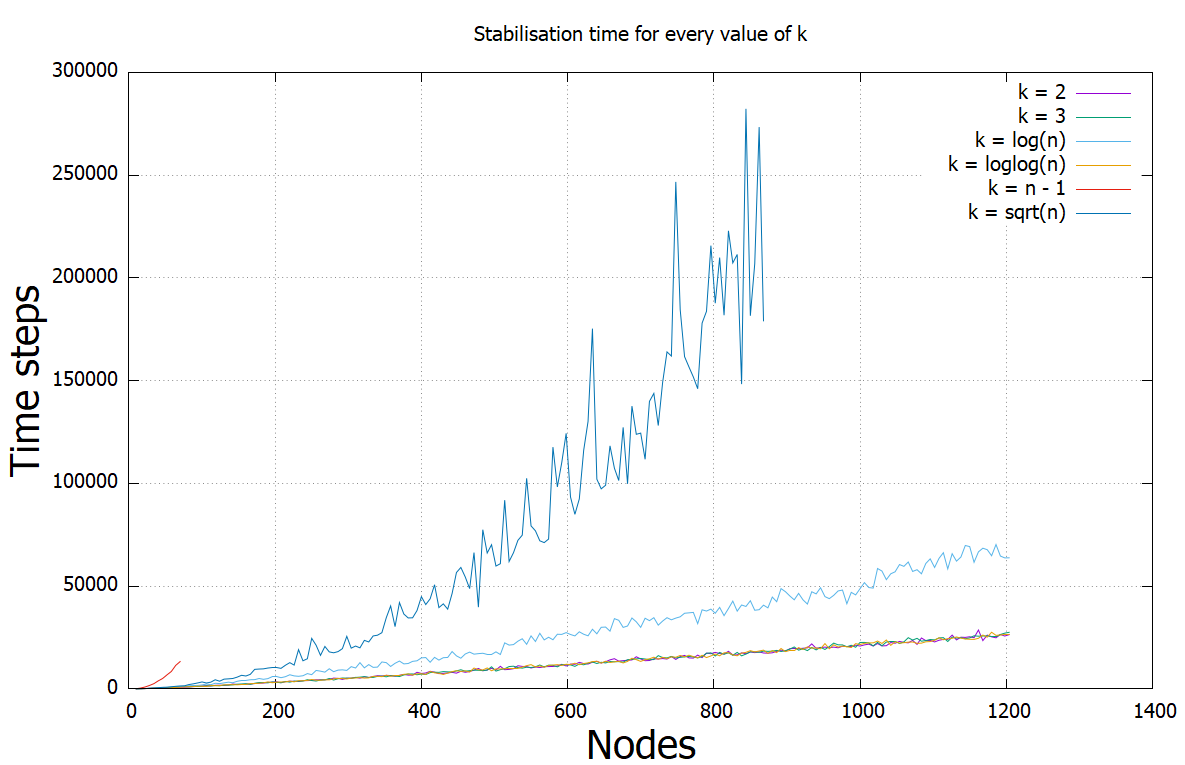}
\caption{The effect of $k$ on the running time of the protocol. A graph which focuses on the lowermost can be found in the Appendix.} \label{fig3}
\end{figure}

\begin{figure}[H]
\includegraphics[width=\textwidth]{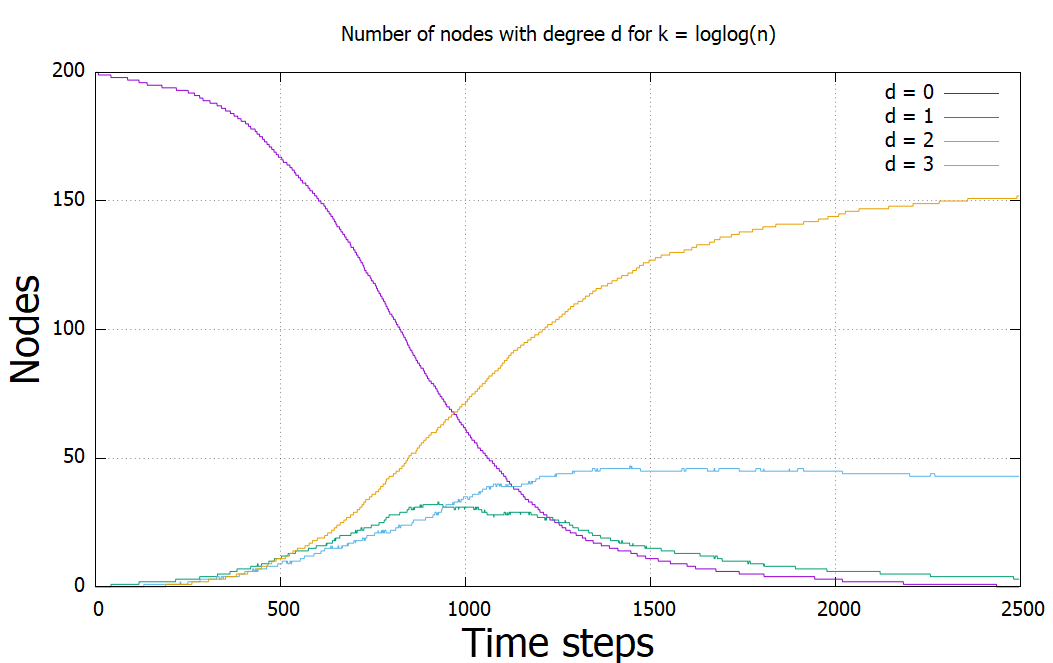}
\caption{The results for $k = \log\log{n}$. Note the difference in the axes labels.} \label{fig5}
\end{figure}

\begin{figure}[H]
\includegraphics[width=\textwidth]{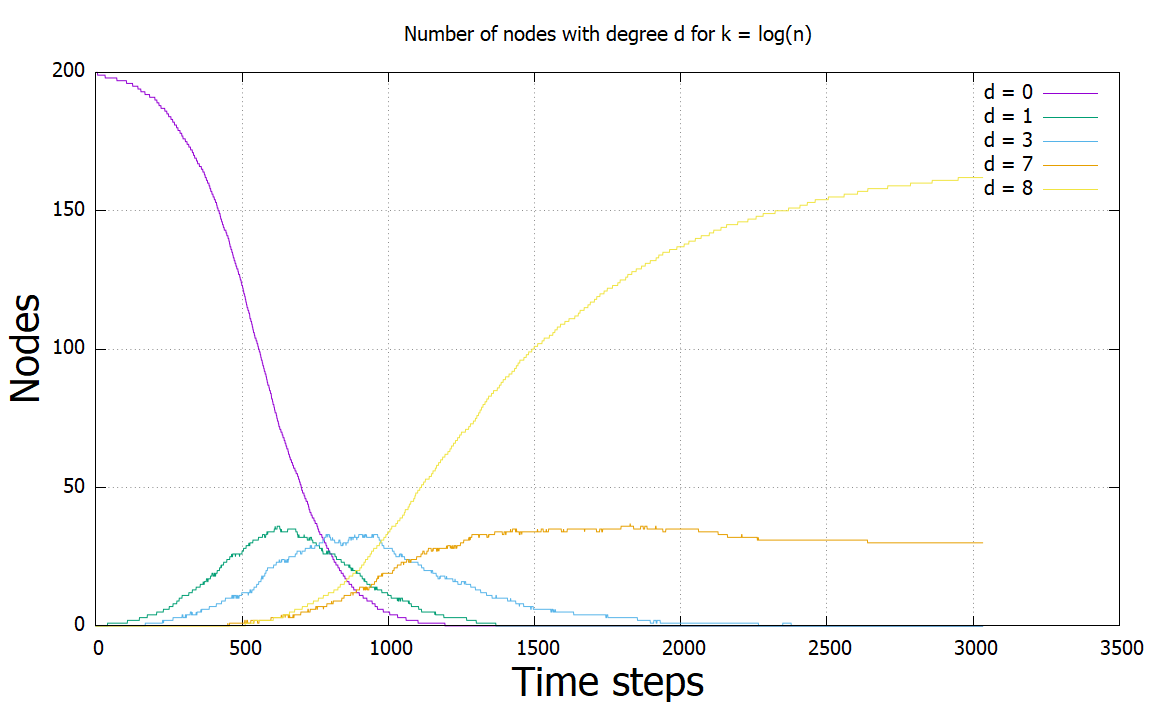}
\caption{The results for $k = \log{n}$. Here we see the beginning of a leftwards shift of the lines, and an upwards shift in $d = k$} \label{fig6}
\end{figure}

\begin{figure}[H]
\includegraphics[width=\textwidth]{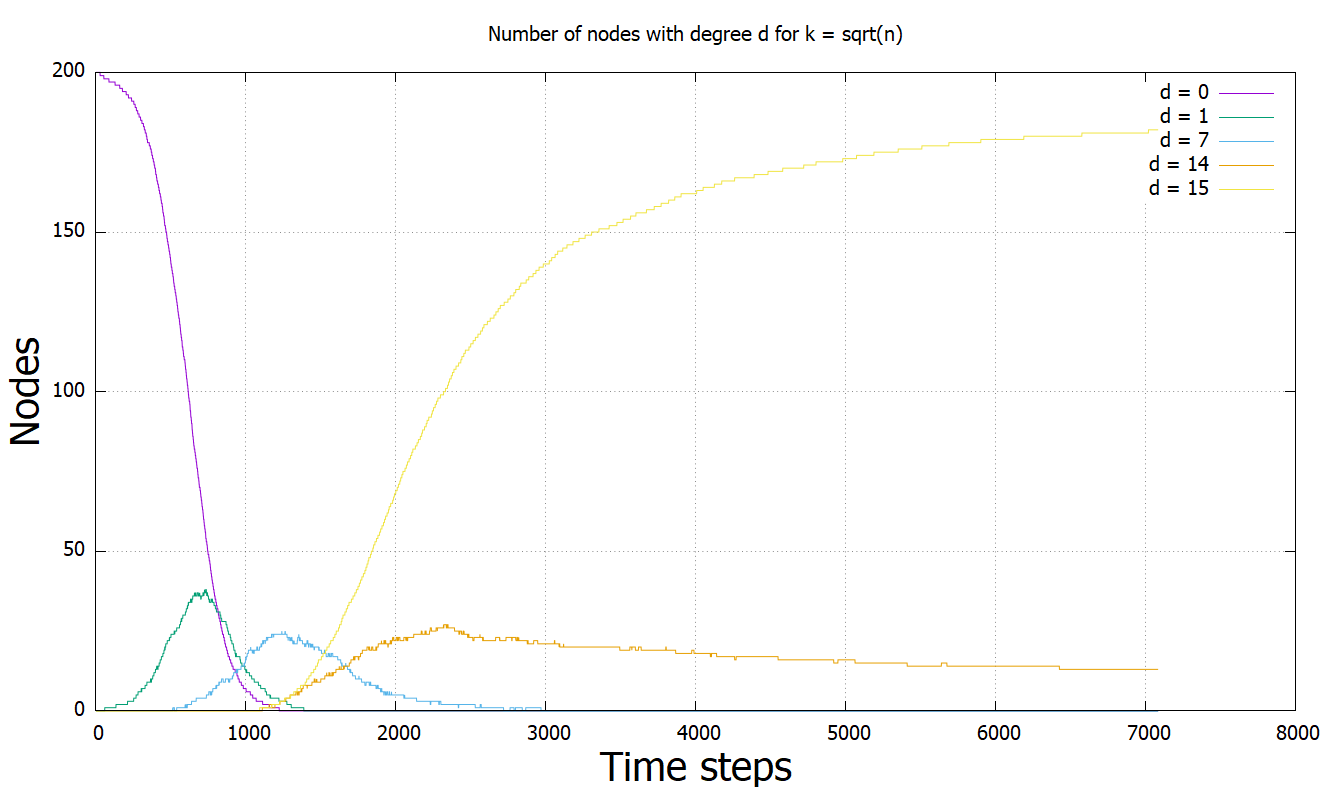}
\caption{The results for $k = \sqrt{n}$. Both shifts are more intense.} \label{fig7}
\end{figure}

To investigate why the protocol slows down dramatically after this point, we ran experiments where we stored the number of nodes with specific degrees in each time step.
We executed the protocol with 200 nodes, and ran 10 iterations.
These degrees were set to $0$, $1$, $k/2$, $k-1$, and $k$.
The results show that the cause seems to be a large reduction in the number nodes which are in the $k-1$ state as $k$ grows as a fraction of $n$.
They suggest that when the fraction of $k-1$ nodes is below some fraction between $1/4$ and $1/8$ of the total the protocol slows down and enters the class of protocols with polynomial time.

\section{Conclusions}
\label{sec:conclusions}

There are a number of open problems to be addressed.
The most important is to develop an exact characterisation of the class of networks which can be constructed in polylogarthimic parallel time.
However, there are other, more immediate problems.
For example, we have yet to investigate the effect that widening the difference between $k$ and $l$ will have on the protocol.
We speculate that this will result in a faster running time in exchange for less uniformity within the resulting spanning network.
We have also speculated about the possibilities of using a leaderless version of our Cross-tree Protocol.
We believe that such a protocol may offer a trade-off between running time and the possibility of forming networks which are spanning, depending on the values of $k$ and $l$.

\newpage

\bibliographystyle{alpha-abr}
\bibliography{algosensors20-sub}

\pagebreak

\appendix

\section*{Appendix}

\section{Notes on the Network Constructor Simulator}

Our paper uses a network constructor simulator to generate results.
It is written in C using CodeBlocks 16.01 as the IDE.
States are generated as a struct containing the ID.
Rules are tuples of the preconditions and effects.
Nodes are also structs with an identifiying index and a pointer to connections stored with allocated memory.\\

The simulator selects two nodes at random, checks for a rule which is applicable and changes the nodes state based on the effects of the rule.
It checks every time step for population stability using a custom condition guaranteed to be fulfiled when the protocol is stable and never otherwise.

\section{Stabilisation Conditions of the $(l, k)$-regular network}

To implement a simulator which can provide results efficiently, we had to define and prove conditions which when fulfilled ensure that the protocol is stable.

\begin{lemma}
For $n > 3$, protocol 3 stabilises with at least 1 node which is not in the state of $k$ in the connected tree.
\end{lemma}

\begin{proof}
If all nodes in the connected tree are in the $O_k$ state, then at some point two $O_{k-1}$ nodes would have to change to the $O_k$ state, which is against the rules of the protocol.
\qed
\end{proof}

\begin{corollary}
Due to the presence of nodes in a state $s \neq O_k$ and the fairness condition, Protocol 3 never stabilises with isolated nodes, defined as nodes which are not part of the main tree structure.
\end{corollary}

\begin{lemma}
For $n > 3$, protocol 3 stabilises with at most $k - 2$ nodes in the states $\{O_x  |  x < k - 1 \}$.
\end{lemma}

\begin{proof}
Assume there are $k - 1$ nodes in the states $O_x$.
If this is the case, there must be some node which is connected to the tree and every other node with any state $O_x$, otherwise the protocol is not stable.
This node must have the state $O_{k - 1}$ which is not in $O_x$.
Therefore the number of nodes in $O_x$ is at most $k - 2$.
\qed
\end{proof}

\begin{theorem}
For $n > k > 3$, at most $n - k - 2$ nodes have a degree of either $k$ or $k - 1$ and $l \leq k - 2$ nodes are of degree at least $1$ and at most $k - 2$.
\end{theorem}

\begin{proof}
All nodes with degree $x < k - 1$ must be members of a clique, otherwise the protocol is not stable.
By Lemma 6, we know there are no isolated nodes, or nodes with degree $d < 1$.
By Lemma 7, we know that there are at most $k - 2$ nodes of degree less than $k - 1$, and that the maximum state within the clique is $O_{k - 2}$.
Threrefore the theorem must hold.
\qed
\end{proof}

\section{Proof of correctness for 2-Slot Protocol}
\label{sec:rotation}

Let $\mathscr{T} = \{G | G \text{ is a tree and } \forall u \in P \implies \Delta^+(u) \leq 2\}$, where $\Delta^+(u)$ is defined as the number of children of the node $u$.

\begin{lemma}
Under protocol 1, the connected component $S$, defined as the leader node and all nodes connected to the leader either directly or indirectly through some other nodes is eventually spanning.
\end{lemma}

\begin{proof}
Let $o$ be the number of open slots in $S$. Formally, $o = \sum_{u \in S} (2 - \Delta^+(u))$.
$o$ is initally 2 as there is one node in $S$ with no children.
$o$ is non-decreasing, as every increase in $\Delta^+(u)$ for some $u$ necessarily increases $|V(S)|$.
Since there are always open slots available, every unconnected node is guaranteed to be able to connect to $S$ at some point.
Therefore when $S$ stabilises it will contain all $u \in P$.
\qed
\end{proof}

\noindent
A node is \emph{available} if it has at least 1 open slot.

\begin{lemma}
For all executions of Protocol 1 on the population $P$ of n nodes, it stabilises to some $G \in \mathscr{T}$ where $|V(G)| = n$.
\end{lemma}

\begin{proof}
We prove this via an induction on the connected component $S$.
For the base case, there is one node in the state $L_0$. This is trivially a member of $\mathscr{T}$ as no connections have formed yet. We now assume that there is a connected component of size $|S|$.
For a connected component of size $|S| + 1$, an unconnected node $u \in V\setminus{S}$ in the state $F$ must connect to $S$ at some node $x \in S$.
By Lemma 1, such a node must exist. If the node $x$ has two children it is in the state $O_2$ or $L_2$, as for all nodes in states $O_i$ and $L_j$ the $i$ and $j$ correspond to the number of children of those nodes.
Since there is no defined transitions from these states no $u$ can connect to $x$.
Therefore $S$ remains a tree and $G(S) \in \mathscr{T}$.
\qed
\end{proof}

\begin{lemma}
For all $G \in \mathscr{T}$, there is an execution of Protocol 1 which stabilises on $G$ when starting on a population $P$ of size $n = |V(G)|$.
\end{lemma}

\begin{proof}
We prove this providing a method to construct any $G \in \mathscr{T}$ with Protocol 1.
Let the leader node $l$ in the population $P$ correspond to the root $r$ of $G$. If $r$ has $i$ children, connect $i$ nodes in the state $F$ to $l$.
For each child $c$ of the leader node, let it correspond to a child $d$ of $r$. If $d$ has $j$ children, connect $j$ nodes in the state $F$ to $c$.
Continuing this process for all nodes $u \in G$, the result is a spanning tree where all nodes in the tree are equivalent to some $u \in G$.
\qed
\end{proof}

\begin{theorem}
Protocol 1 stably constructs the graph language $\mathscr{T}.$
\end{theorem}

\begin{proof}
By application of the Lemmas above.
\end{proof}

\section{Expected Running Time of 2-Children Spanning Tree}

\begin{lemma}
Let $T \in \mathscr{T}$ of $n$ nodes. The number of available nodes $\alpha(T) = \floor*{|T|/2}+1$.
\end{lemma}

\begin{proof}
Observe that for $T$, every second node which connects to $T$ keeps the number of available nodes the same.
This is because two new nodes must become children of the same node, and the second new node takes the second slot.
For the base case, $n = 1$ and $\alpha = 0 + 1 = 1$.
We divide $n = n + 1$ into two cases: $n$ is \emph{even} and $n$ is \emph{odd}.
If $n$ is \emph{even}, then $\alpha = n/2 + 1$.
Then for $n = n+1$, $\alpha = \floor*{n+1/2} + 1 = n/2 + 1$. This corresponds to observation earlier that every other node (i.e $n$ is odd) should not increase $\alpha$.
If $n$ is \emph{odd}, then $\alpha = \floor*{n/2} + 1 = (n-1)/2 + 1$.
Then for $n = n+1$, $\alpha = \floor*{n+1/2} + 1 =  n/2 + 1$ as expected.

\noindent
\emph{Remark:}
At any point during the execution of $A$, for the connected component $S$, $G(S) \in \alpha(T)$.
\qed
\end{proof}

Let the probablistic process $\mathscr{P}$ be an execution of the protocol 1 with the following scheduling restriction:
If at any point during the execution of  $A$ two nodes $x$ and $y$ have exactly one child, disconnect that child of $x$ or $y$ which is a leaf and connect it to the other node.
If both are leaves, pick one at random.\\

\begin{lemma}
The expected time to convergence of the probabalistic process $\mathscr{P}$ is $O(\log{n})$.
\end{lemma}

\begin{proof}
Let the r.v. $X$ be the number of steps until convergence.
A step is \emph{successful} if any unconnected node joins the connected component $S$.

\noindent
An \emph{epoch} $i$ is the period beginning with the step following the $(i-1)$st success and ending with the step at which the $i$th success occurs.
The r.v. $X_i$, $1 \leq i \leq n-1$, is the number of steps in epoch $i$.

\noindent
$p_i$ is the probability of success at any step in epoch $i$. This is defined as $p_i = \frac{2\alpha(T_i)(n-i)}{n(n-1)}$, where $T_i$ is the graph of the strongly connected component G(S) in epoch $i$.\\
It follows that $E[X_i] = 1/p_i = \frac{n(n-1)}{2\alpha(T_i)(n-i)}$.

\noindent
By linearity of expectation we have

\begin{align*}
E[X] &= E\Bigg[\sum_{i=1}^{n-1} X_i\Bigg] = \sum_{i=1}^{n-1} E[X_i] = \sum_{i=1}^{n-1} \frac{n(n-1)}{2\alpha(T_i)(n-i)}
= \frac{n(n-1)}{2}\sum_{i=1}^{n-1}\frac{1}{\alpha(T_i)(n-i)}\\
&= \frac{n(n-1)}{2}\sum_{i=1}^{n-1}\frac{1}{(\floor*{i/2} + 1)(n-i)}
\leq \frac{n(n-1)}{2}\sum_{i=1}^{n-1}\frac{1}{(i/2)(n-i)}\\
&= n(n-1)\sum_{i=1}^{n-1}\frac{1}{i(n-i)}
= n(n-1)\sum_{i=1}^{n-1}\frac{1}{n}\Bigg(\frac{1}{i} + \frac{1}{n-i}\Bigg)\\
&= (n-1)\Bigg[\sum_{i=1}^{n-1}\frac{1}{i} + \sum_{i=1}^{n-1}\frac{1}{n-i}\Bigg]
= (n-1)2H_{n-1} = 2(n-1)[\ln(n-1) + O(1)]\\
&= O(n\log{n})
\end{align*}
\qed
\end{proof}

\begin{lemma}
The running time of $\mathscr{P}$ is the worst case running time for the protocol 1.
\end{lemma}

\begin{proof}
Assume there is an execution of $A$ which has an slower running time than $\mathscr{P}$.
Such an execution must have a lower number of available nodes at some point than $\mathscr{P}$.
If the execution simulates the scheduling restriction of $\mathscr{P}$ then it cannot be slower than $\mathscr{P}$.
If the execution does not simulate the restriction then at some point two nodes have two leaves and one is not shifted to the other.
The number of available nodes is therefore greater by one and the expected running time faster than $\mathscr{P}$.
Therefore any execution of $A$ must be at least as fast as $\mathscr{P}$.
\qed
\end{proof}

\begin{theorem}
The expected running time of protocol 1 is upper bounded by the O($\log{n}$) running time of $\mathscr{P}$.
\end{theorem}

\begin{proof}
By application of lemmas 13 and 14.
\qed
\end{proof}

\section{$k = 3$ Formal Proof Strategy}

To investigate why the running time of the Protocol is Polylogarithmic, we modified the simulator to perform a single test and output the degree of each node for every time step.

\begin{figure}[H]
\includegraphics[width=\textwidth]{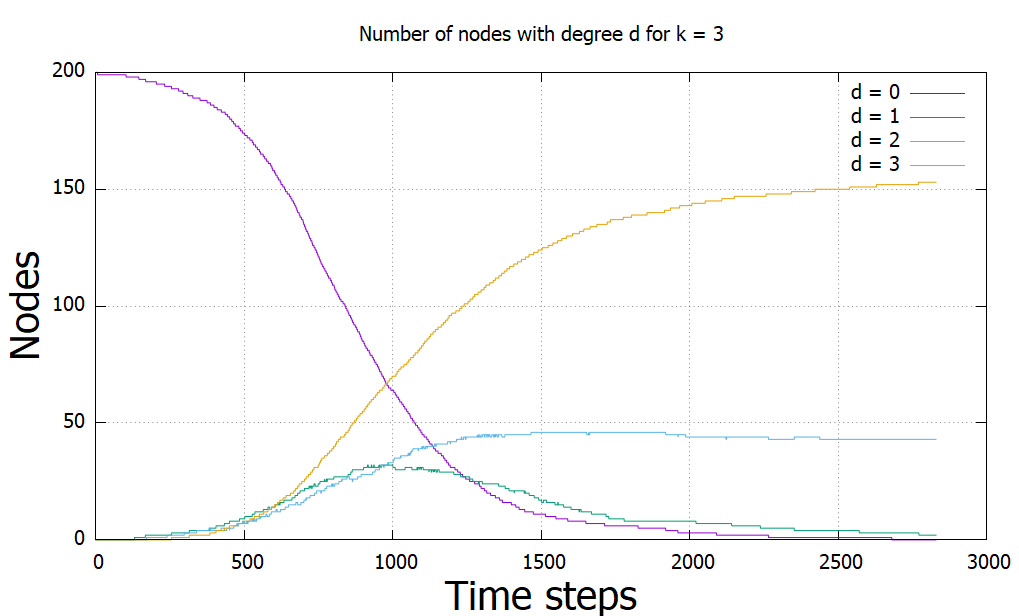}
\caption{Nodes by degree for the protocol when k = 3} \label{fig2}
\end{figure}

Based on the above we have created the following strategy for proving that the running time of Protocol 3 for $k = 3$.\\

The proof is divided into two phases. 
In phase 1, the number of nodes with degree 0 is large, but after some time it will be small.
It can be shown that in the time it takes for this to happen, the number of nodes with degrees 2 are at least some fraction of $n$ with high probability, perhaps $n/A$ for some constant $A$.
In phase 2, which begins when the number of degree 0 nodes is small, it can be shown that the number of degree 2 nodes remains at least $n/A$ w.h.p and that this allows the protocol to stabilise in polylogarithmic time for arbitrarily low numbers of degree 0/1 nodes.

\newpage

\section{Running time for $k < \log{n}$}

\begin{figure}[h!]
\includegraphics[width=\textwidth]{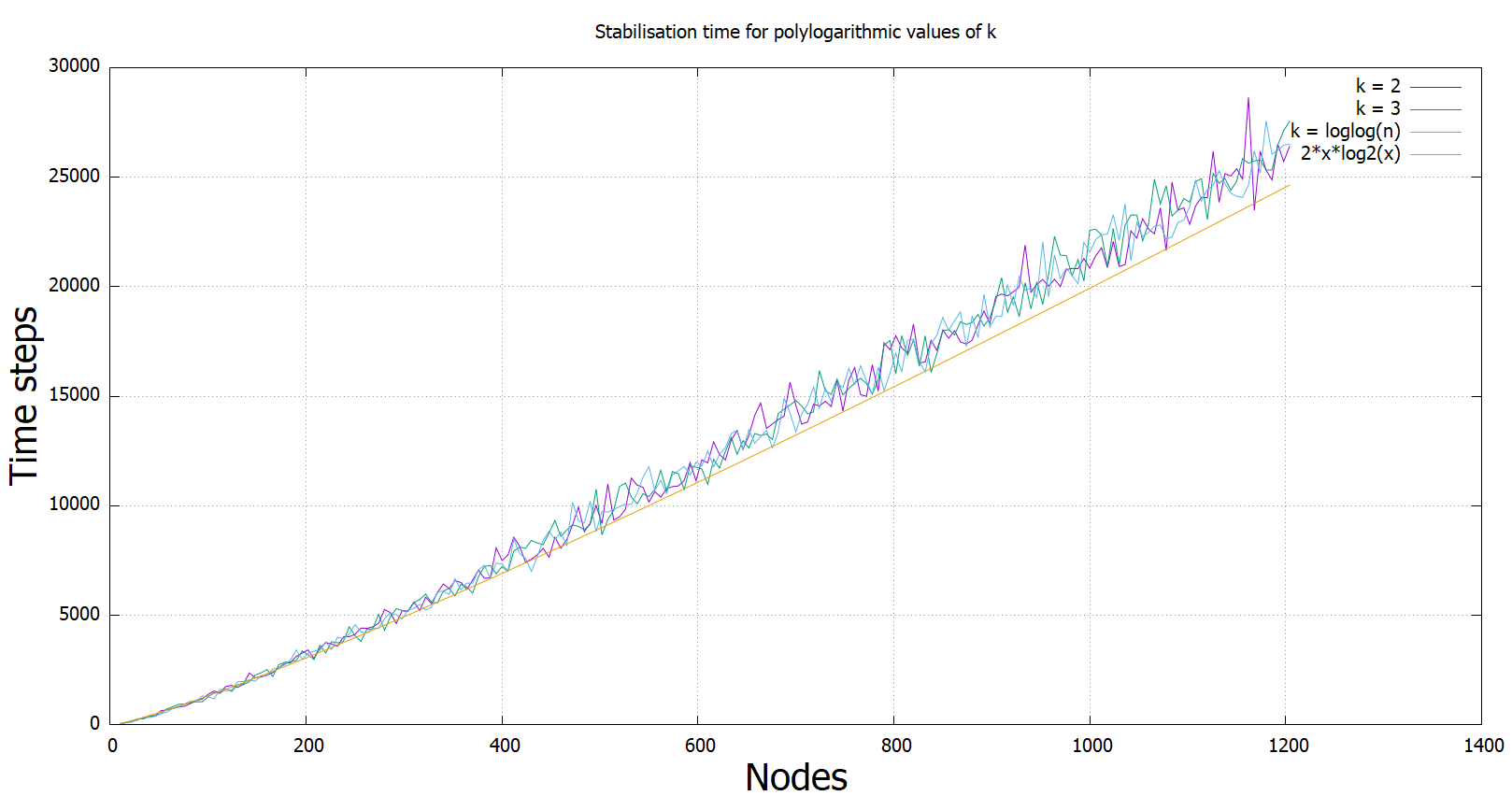}
\caption{The polylograthmic results in detail.} \label{fig4}
\end{figure}

\end{document}